\documentclass[11pt,a4paper]{amsart}
\usepackage{amsmath,amsthm,amssymb,amscd,mathrsfs,dsfont,hyperref,a4wide}
\newcommand{\ud}{\mathrm{d}}
\newcommand{\mf}{\mathfrak}
\newcommand{\ui}{\mathrm{i}}
\newcommand{\ue}{\mathrm{e}}
\newcommand{\eins}{\mathds{1}}
\newcommand{\be}{\begin{equation}}
\newcommand{\ee}{\end{equation}}
\newcommand{\ba}{\begin{aligned}}
\newcommand{\ea}{\end{aligned}}

\newcommand{\rz}{{\mathbb R}}
\newcommand{\nz}{{\mathbb N}}
\newcommand{\kz}{{\mathbb C}}
\newcommand{\bs}{\boldsymbol}
\DeclareMathOperator{\asinh}{arsinh}
\DeclareMathOperator{\Or}{O}
\DeclareMathOperator{\bv}{bv}
\DeclareMathOperator{\reg}{reg}
\DeclareMathOperator{\im}{Im}
\DeclareMathOperator{\re}{Re}
\DeclareMathOperator{\tr}{Tr}      
\DeclareMathOperator{\supp}{supp}
\DeclareMathOperator{\sgn}{sgn}
\newtheorem{theorem}{Theorem}[section]
\newtheorem{lemma}[theorem]{Lemma}
\newtheorem{prop}[theorem]{Proposition}

\newtheorem{remark}[theorem]{Remark}

\newtheorem{defn}[theorem]{Definition}
\newtheorem{rem}[theorem]{Remark}
\begin{document}
\title[Heat-kernel trace of a singular two-particle 1-d contact interaction]{An asymptotic expansion of the trace of the heat kernel of a singular two-particle contact interaction in one dimension}
\author{Sebastian Egger}
\address[S.~Egger]{Department of Mathematics, Technion-Israel Institute of Technology 629 Amado Building, Haifa 32000, Israel}
\email{egger@tx.technion.ac.il}
\begin{abstract} 
The regularized trace of the heat kernel of a one-dimensional Schr\"odinger operator with a singular two-particle contact interaction being of Lieb-Liniger type is considered. We derive a complete small-time asymptotic expansion in (fractional) powers of the time, $t$. Most importantly, we do not invoke standard parametrix constructions for the heat kernel. Instead, we first derive the large-energy expansion of the  regularized trace of the resolvent for the considered operator. Then, we exploit that the resolvent may be obtained by a Laplace transformation of the heat semi-group, and an application of a suitable inverse Watson lemma eventually yields the small-$t$ asymptotic expansion of the heat-kernel trace.      
\end{abstract}
\maketitle
\section{Introduction}
\label{intro}
The classical heat kernel is the kernel of an integral operator generating the semi-group of the heat equation on domains or manifolds. Given sufficient regularity of the boundary and assuming that the volume of the domain is finite one may conclude that the heat semi-group is a trace class operator. Mercer's famous theorem then tells that its trace may be calculated by integrating the heat kernel along its diagonal. 

Minakshisundaram and Pleijel showed in their celebrated work, \cite{MR0032861,MR0031145}, that the classical heat kernel possesses a complete small-time (small-$t$) asymptotic expansion in (fractional) powers of $t$. They were also able to show that its coefficients bear topological and metric information of the underlying domain or manifold. Most notable here is the heat-kernel approach to the profound Atiyah-Singer index theorem, \cite{MR0324731,MR0650829,MR0650828}, which links a specific coefficient of the heat-kernel expansion to the topological and analytic index, respectively, of an elliptic operator acting on sections of a corresponding vector bundle.

Another and very interesting aspect pointed out in \cite{MR0032861,MR0031145} is that the small-$t$ asymptotics of the trace of the heat kernel and the meromorphic properties of the spectral $\zeta$-function corresponding to the generating Laplacian of the heat equation are very closely related. From a physical point of view, the spectral $\zeta$-function is an important object in statistical physics and quantum mechanics to calculate, e.g., path integrals (Brownian motion), \cite{MR0524257}, spectral determinants in lattice QCD, or the Casimir force in QED, see, e.g., \cite{MR1890614,MR1730139}.
 
A standard approach to find an asymptotic expansion of the trace of the heat kernel is to derive a parametrix, i.e., a local small-$t$ asymptotic expansion for the solution of the heat equation, and then integrate its position dependent coefficients along their diagonals. The advantage of this method is that it's rather universally applicable and the coefficients may be identified as \textit{local invariants} of the underlying domain or manifold. Due to the overwhelming number of remarkable results we refer here (and references therein) to \cite{MR1130604,MR1994690,MR2017528,MR2064001,MR2040963,MR2550208,MR3137026} for the heat-kernel expansion in particular and to \cite{MR1780769,MR2569498,st09,MR2641770,MR3077281,MR3157991,KM18,MR3770367} for related modern results regarding asymptotic properties of the heat kernel. 

On metric graphs the trace of the heat kernel was first studied by Roth, \cite{MR711833}, deriving an exact (Selberg-like) trace formula for it. After that, the (one-particle) heat kernel has experienced an accelerated attention and analogous questions has been asked to the manifold case, \cite{MR1207868,MR1227664,MR1853353,MR2423580,MR2291812,MR3034773,MR3243602,MR3335059,MR3394421}. The motivation of studying contact interactions in many-particle physics for one-dimensional systems dates back to the famous Lieb-Liniger model, \cite{PhysRev.130.1605}, used to test Bogoliubov's perturbation theory for Bose gases. Since then many-particle contact interactions are well established in one dimensions and graphs for testing and modeling famous phenomenons such as superconductivity or Bose-Einstein condensation, \cite{MR3050159,MR3204513,MR3487209,MR3381152,MR3227513,2018arXiv180500725B,MR3720512}. For the importance of one-dimensional contact interactions in physics and recent experimental implementations we refer to \cite{RevModPhys.83.1405}. Those exciting results inspired us to consider as a very first example the small-$t$ asymptotics of the trace of the heat kernel for a simple but non-standard Lieb-Liniger type system of a singular but non-constant contact interaction on the real line, $\rz$. 

More precisely, the operator which we consider is the one-dimensional two-particle Schr\"odinger operator given by the formal expression
\be
\label{schroe}
-\Delta_{\rho}:=-\partial^2_{x_1}-\partial^2_{x_2}+\rho(x_1,x_2)\delta(x_1-x_2)\,,
\ee
with $L^2(\rz^2)$ as the two-particle Hilbert space and a $\delta$-potential modulated by a potential $\rho$. We assume that the modulating potential, $\rho$, is smooth and possesses compact support. The heat semi-group, $\ue^{t\Delta_{\rho}}$, satisfies $\partial_t\ue^{t\Delta_{\rho}}|_{t=0}=\Delta_{\rho}$ in a strong sense, and the heat kernel, $k^\rho(t)(\cdot,\cdot)$, $t\in\rz$, is the integral kernel generating $\ue^{t\Delta_{\rho}}$ by 
\be
(\ue^{t\Delta_{\rho}}\psi)(\bs{x})=\int\limits_{\rz^2}k^{\rho}(t)(\bs{x},\bs{y})\psi(\bs{y})\ud \bs{y}\,.
\ee
However, since the underlying configuration space is non-compact the Schr\"odinger operator \eqref{schroe} possesses an essential spectrum, and we have to regularize the heat semi-group in order to make it a trace-class operator. 

In this paper, we follow the method of \cite{MR3335059} and we don't derive a parametrix expansion for the heat kernel. Instead, we exploit that the resolvent and the heat semi-group are related via a Laplace transformation. This approach is similar to \cite{MR1994690}, where the authors used the so-called Agmon-Kannai method to derive an asymptotic expansion for the resolvent kernel of the corresponding Schr\"odinger operator. The Agmon-Kannai method, in turn, is a tool to obtain a series representation of the resolvent with operator-valued coefficients and is based on a recursive construction involving commutators of the free and the full resolvent of the considered operator, \cite{MR1809669}.

We don't use the Agmon-Kannai method here, but we exploit that the resolvent of our system allows a rather explicit representation in terms of suitable integral operators which is based on a more-general formula of Kre\u{\i}n. We start by first establishing an asymptotic expansion for large but negative energies of the regularized trace of the resolvent. Then, we use a suitable version of the converse Watson lemma to deduce the small-$t$ asymptotics of the regularized trace of the heat semi-group (and heat kernel). In general, our method is taking advantage of the symmetry properties of the underlying system, and the convenience of our method is that it allows a compact, explicit and quick formulation of the heat-kernel coefficients.

Finally, we refer to Appendix \ref{not}, where we recall standard notations and definitions used in this paper.

\section{Preliminaries}
\label{prel}
We begin our investigations by implementing a rigorous version of the formal Schr\"odinger operator \eqref{schroe}. To obtain a well-defined operator in $L^2(\rz^2)$ it is convenient to associate with \eqref{schroe} a quadratic form being complete and semi-bounded from below. Then, there is a unique self-adjoint operator corresponding to the quadratic form, \cite[Section~4.2]{Weidmann:2000}, which may be regarded as a rigorous version of \eqref{schroe}. 

To get the quadratic form of the operator \eqref{schroe} we first observe that the modulating potential $\rho$ only needs to be known on the diagonal (one-dimensional submanifold)
\be
D:=\{(x,x): \ x\in\rz\}\subset\rz^2\,
\ee
due to the $\delta$-potential interaction. Now, using the identification $\rho(x,x)={\rho}(x)$ we restrict ourselves to compactly supported and smooth potentials, i.e., ${\rho}\in C^{\infty}_{0}(\rz)$. To identify the associated quadratic form and operator we proceed as in \cite[Section~3.1]{Kerner}, replacing the interval $[0,1]$ with $\rz$ to obtain our case related to \eqref{schroe}. We follow the steps of \cite[Section~3.1]{Kerner} and to do so, we use $D$ and partition $\rz^2$ into the two disjoint open sets $D_+$ and $D_-$ by  
\be
\label{Dr2}
\rz^2=D_-~\dot{\cup}~D~\dot{\cup}~D_+\,,\quad D_+=\{(x_1,x_2):\ x_1<x_2\}\,, \quad D_-:=\{(x_1,x_2):\ x_1>x_2\}\,.
\ee
Moreover, since $D$ is a straight line, and hence a smooth curve, we may define the trace maps see, e.g., \cite[Theorem~1.5.1.1]{MR775683}
\be
\bv_{\pm}:H^1(D_{\pm})\rightarrow H^{\frac{1}{2}}(D). 
\ee
Both above maps are continuous linear maps. In the same way, the gradients 
\be
\nabla: H^2(D_{\pm})\rightarrow H^{1}(D_{\pm})\oplus H^{1}(D_{\pm})
\ee
are well-defined and continuous maps. Therefore, the inward normal derivatives, $\partial_{\bs{n}}$, w.r.t. the boundary $D$ of the domains $D_{\pm}$ act as
\be
\partial_{\bs{n}}:H^2(D_{+})\oplus H^2(D_{-})\rightarrow H^{\frac{1}{2}}(D)\oplus H^{\frac{1}{2}}(D)\,,
\ee
and are well-defined by
\be
\partial_{\bs{n}}(\psi_+\oplus\psi_-):=\psi_{\bs{n},+}\oplus\psi_{\bs{n},-}\,,
\ee
and
\be
\psi_{\bs{n},\pm}:=\mp(\partial_{x_1}-\partial_{x_2})\psi_{\pm}\,.
\ee

With these technical tools at hand, we may now use \cite[p.~6]{Kerner} allowing a proper identification of the two-particle and one-dimensional Schr\"odinger operator in \eqref{schroe} with a one-particle and two-dimensional operator acting on $\rz^2$. For the readers convenience we denote this operator by $-\Delta_{\rho}$ as well. The functions of the corresponding operator domain, $D(\Delta_{\rho})$, obey the following regularity and boundary conditions. If $\psi\in D(\Delta_{\rho})$ then $\psi \in H^2(D_+)\oplus H^2(D_-)\subset L^2(\rz^2)$ and the following boundary conditions are satisfied in a $L^2$-sense, \cite[p.~6]{Kerner}:
\be
\bv_+\psi=\bv_-\psi=:\Psi\,, \quad \psi_{\bs{n},+} + \psi_{\bs{n},-}=\rho\Psi\,.
\ee
Moreover, by \cite[p.~6]{Kerner} the operator $-\Delta_{\rho}$ is associated with the quadratic form $(q_{\rho},H^1(\rz^2))$ defined by
\be
q_{\rho}(\psi):=\int\limits_{\rz^2}\langle\nabla\psi,\nabla\psi\rangle_{\rz^2}\ud \bs{x}+\int\limits_{D}\rho\Psi\ud x\,,
\ee
where we used \eqref{Dr2}. On the other hand, given $(q_{\rho},H^1(\rz^2))$ then the associated operator is $-\Delta_{\rho}$ and is self-adjoint and bounded from below, \cite[Theorem~4.2]{MR1275948}. We denote by $\lambda_{\min,\rho}:=\inf\{\lambda\in\rz:\ \lambda\in\sigma(-\Delta_\rho)\}$ the bottom of the spectrum $\sigma(-\Delta_\rho)$ of the Schr\"odinger operator, and we recall the well-known fact $\lambda_{\min,0}=0$.

\section{The resolvent kernel}
At the beginning of this section, we derive an explicit expression of the resolvent, $\lambda\notin\sigma(-\Delta_{\rho})$,
\be
\label{res}
R_{\rho}(\lambda):=(-\Delta_{\rho}-\lambda)^{-1}\,.
\ee
To do so, we follow \cite{MR1275948}, and in the following, we choose $\sqrt{\lambda}=k\in\kz$ such that $\im k>0$.  We introduce the integral kernels
\be
\ba
G_0(k)(x_1,x_2,y_1,y_2)&:=\frac{1}{2\pi}K_0(-\ui k\sqrt{(x_1-y_1)^2+(x_2-y_2)^2})\,,
\ea
\ee
and
\be
\label{gk}
g(k)(x,y):=\frac{1}{2\pi}K_0(-\ui \sqrt{2} k|x-y|)\,.
\ee
The integral kernel $G_0(k)$ corresponds for $\im k> 0$ to a bounded operator $R_0(k):L^2(\rz^2)\rightarrow L^2(\rz^2)$
and $g(k)$ to a bounded operator $\mf{g}(k):L^2(\rz)\rightarrow L^2(\rz)$, respectively. That may be deduced from the asymptotic behavior of the $K_0$-Bessel function for large arguments, \cite[p.~139]{Magnus:1966}. It is worth mentioning that the logarithmic singularity of the $K_0$-Bessel function at the origin, \cite[p,~65]{Magnus:1966}, doesn't affect the boundedness of $R_0$ and $\mf{g}$ due to Young's inequality. Note that $R_0(k)$ is the (free) resolvent of the pure Laplacian on $\rz^2$. To write down the resolvent explicitly we need one more integral operator connecting $L^2(\rz^2)$ and $L^2(\rz)$. We introduce
\be
\label{b}
b(k)(x,y_1,y_2):=\frac{1}{2\pi}K_0(-\ui k\sqrt{(x-y_1)^2+(x-y_2)^2})\,,
\ee
and those integral kernel generates for $\im k>0$ a bounded operator $\mf{b}(k):L^2(\rz^2)\rightarrow L^2(\rz)$. Finally, we make the simple observation that any potential $\rho\in C^{\infty}_0(\rz)$, or only $\rho\in L^\infty(\rz)$, generates a bounded multiplication operator $\rho:L^2(\rz)\rightarrow L^2(\rz)$.

With these operators at hand, we now invoke \cite[Corollary~2.1]{MR1275948} saying that the resolvent $R_{\rho}(\lambda)$ may be written as
\be
\label{resolvent}
R_{\rho}(\lambda)={R}_0(k)-\mf{b}\left(\overline{k}\right)^{\ast}(\eins+\rho\mf{g}(k))^{-1}\rho{b}(k)\,,
\ee
where we put $k=\sqrt{\lambda}$.
In the following, it is convenient to denote by $C_{\alpha}$, $\alpha<\frac{\pi}{2}$, the cone around the positive imaginary axis $\ui\rz^+$ with opening angle $2\alpha$, i.e.,
\be
\label{calpha}
C_{\alpha}:=\left\{z\in\kz:\ |\arg(z)-\frac{\pi}{2}|<\alpha\right\}\,.
\ee
Then, for $k\in C_{\alpha}$, $\alpha<\frac{\pi}{2}$, and $|k|$ large enough one has, \cite[Corollary~2.2]{MR1275948},
\be
\label{smallerone}
\left\|\rho\mf{g}(k)\right\|<1\,.
\ee
Now, we introduce the regularized resolvent $R^{\reg}_{\rho}(k)$, $k\in C_{\alpha}$, $\alpha<\frac{\pi}{2}$, and $|k|$ sufficiently large, defined as $R^{\reg}_{\rho}(k):=R_{\rho}(k^2)-R_{0}(k^2)$. Due to the semi-boundedness of $q_{\rho}$ (and $q_{0}$) the operator $R^{\reg}_{\rho}(k)$ exists for $k\in C_{\alpha}$, $\alpha<\frac{\pi}{2}$ and $|k|$ sufficiently large. We will show that $R^{\reg}_{\rho}(k)$ is also a trace class operator. For this, it is advantageous to 'shift' a square root of the potential in \eqref{resolvent} from right to left. This will eventually reveal that the supports of the integral kernels are compact w.r.t. appropriate variables, and will be exploited to estimate the integrals from above. Specifically, the following rearrangement of $R^{\reg}_{\rho}(k)$ is possible. 
\begin{lemma}
For $k\in C_{\alpha}$, $\alpha<\frac{\pi}{2}$, and $|k|$ sufficiently large, we have
\be
\label{betterres}
R^{\reg}_{\rho}(k)=-\left(\sqrt{|\rho|}\mf{b}\left(\overline{k}\right)\right)^{\ast}(\eins+\sgn{\rho}{\sqrt{|\rho|}}\mf{g}(k)\sqrt{|\rho|})^{-1}\sgn{\rho}{\sqrt{|\rho|}}\mf{b}(k)\,.
\ee
\end{lemma}
\begin{proof}
Using \eqref{smallerone} we may write $(\eins+\rho\mf{g}(k))^{-1}$ as a Neumann series. Now, 
\be
(\rho\mf{g}(k))^n\rho=\sqrt{|\rho|}(\sgn{\rho}{\sqrt{|\rho|}}\mf{g}(k)\sqrt{|\rho|})^n\sgn{\rho}{\sqrt{|\rho|}}
\ee
for every $n\in\nz_0$ proves the claim.
\end{proof}
\subsection{Trace class property of the regularized resolvent.}
Now, we are in the position to prove that the regularized resolvent is actually a trace class operator.
\begin{prop}
\label{restraceclass}
For $k\in C_{\alpha}$, $\alpha<\frac{\pi}{2}$, and $|k|$ sufficiently large, $R^{\reg}_{\rho}(k)$ is a trace class operator.
\end{prop}
\begin{proof}
We want to employ \cite[Satz~3.23]{Weidmann:2000} saying that it's enough to show that $R^{\reg}_{\rho}(k)$ can be factorized as 
\be
R^{\reg}_{\rho}(k)=AB
\ee
by two Hilbert-Schmidt operators $A:L^2(\rz)\rightarrow L^2(\rz^2)$ and $B:L^2(\rz^2)\rightarrow L^2(\rz)$. Looking at \eqref{betterres} we are tempted to identify 
\be
A=-\left(\sqrt{|\rho|}\mf{b}\left(\overline{k}\right)\right)^{\ast}\,,
\ee
and 
\be
B=(\eins+\sgn{\rho}{\sqrt{|\rho|}}\mf{g}(k)\sqrt{|\rho|})^{-1}\sgn{\rho}{\sqrt{|\rho|}}\mf{b}(k)\,.
\ee
Indeed, by \cite[Satz~3.18]{Weidmann:2000} the Hilbert-Schmidt property is closed under taking the adjoint. Moreover, by \cite[Satz~3.20]{Weidmann:2000} we may neglect the operator $(\eins+\sgn{\rho}{\sqrt{|\rho|}}\mf{g}(k)\sqrt{|\rho|})^{-1}\sgn{\rho}$, and it remains to show that $\sqrt{|\rho|}\mf{b}\left({k}\right)$ is of Hilbert-Schmidt class. We proceed to show that the integral kernel of $\sqrt{|\rho|}\mf{b}\left({k}\right)$ is in $L^2(\rz\times\rz^2)$ which then proves by \cite[Satz~3.19]{Weidmann:2000} the claim. We recall \eqref{b} and calculate 
\be
\label{inegral}
\ba
&(2\pi)^2\int\limits_{\rz}\int\limits_{\rz^2}|\rho(x)||{b}\left({k}\right)(x,y_1,y_2)|^2\ud x\ud y_1\ud y_2\\
&\hspace{0.25cm}=\int\limits_{\supp\rho}\int\limits_{B_R(\bs{0})}|\rho(x)||{K_0}\left(-\ui k\sqrt{(x-y_1)^2+(x-y_2)^2}\right)|^2\ud x\ud y_1\ud y_2\\
&\hspace{0.5cm}+\int\limits_{\supp\rho}\int\limits_{\rz^2\setminus B_R(\bs{0})}|\rho(x)||{K_0}\left(-\ui k\sqrt{(x-y_1)^2+(x-y_2)^2}\right)|^2\ud x\ud y_1\ud y_2\,,
\ea
\ee
where $B_R(\bs{0})$ is a ball with sufficient large radius $R$ such that 
\be
x\in\supp\rho \quad \Rightarrow \quad y_1\neq x \quad \text{or} \quad y_2\neq x
\ee
holds in $\rz^2\setminus B_R(\bs{0})$. Such a radius, $R$, exists as the support of $\rho$ is bounded. By the same reason and due to the asymptotic behavior \eqref{smallz} we have that
\be
\label{firstes}
\ba
&\int\limits_{\supp\rho}\int\limits_{B_R(\bs{0})}|\rho(x)||{K_0}\left(-\ui k\sqrt{(x-y_1)^2+(x-y_2)^2}\right)|^2\ud x\ud y_1\ud y_2\\
&\hspace{0.25cm}<C\int\limits_{\supp\rho}\int\limits_{B_R(\bs{0})}|\ln(kr)|^2r\ud r\ud x<\infty
\ea
\ee
with some $C>0$. By \eqref{largez} a similarly estimates to \eqref{firstes} also holds for the second integral on the r.h.s. in \eqref{inegral} proving the claim. 
\end{proof}
In order to evaluate the trace of the heat kernel, we want to employ a generalized version of Mercer's theorem. To achieve this, we first have to show that the heat kernel exists and then to work out suitable continuity and decay properties satisfied by that integral kernel. 
\begin{lemma}
\label{lemma1}
Given $k\in C_{\alpha}$, $\alpha<\frac{\pi}{2}$, and $|k|$ sufficiently large, $R^{\reg}_{\rho}(k)$ is an integral operator with an continuous and exponentially decaying kernel for large arguments. 
\end{lemma}
\begin{proof}
First, we choose for our convenience an appropriate $-\ui k\in\rz^+$, and the remaining cases may then be similarly proven. With a similar argument as in the proof of Proposition~\ref{restraceclass}, using the (absolute) integrability of the logarithm, we may regard $\sqrt{|\rho|}\mf{b}(k)$ as a map 
\be
\tilde{b}:(\rz^2,\|\cdot\|_{\rz^2})\rightarrow L^2(\rz)\,,
\ee
defined by 
\be
\label{tildeb}
(\tilde{b}(y_1,y_2))(x):=\sqrt{|\rho(x)|}b(k)(x,y_1,y_2)\,.
\ee
To see that $\tilde{b}$ is continuous we first take into account that the support of $\rho$ is finite and therefore we only have to investigate singular points of \eqref{tildeb} w.r.t. the argument. Looking at \eqref{b} we see that the singularity is in the logarithm of \eqref{smallz} where $x=y_1=y_2$. We may choose $y_1=y_2=0$, $y_1'=y_2'=:y'>0$ and the other cases are similar. We are going to use that for every $\epsilon>0$ there is a $\delta(\epsilon)$ such that 
\be
\ba
&\|(0,0)-(y',y')\|_{\rz^2}\leq\delta(\epsilon), \ \text{and} \ (x,x)\notin B_{2\delta(\epsilon)}(0,0)\\
&\hspace{0.25cm}\Rightarrow \ |\ln\|(x,x)\|_{\rz^2}-\ln\|(y'-x,y'-x)\|_{\rz^2}|\leq\epsilon\,.
\ea
\ee
Hence, choosing $\epsilon$ suitable small gives
\be
\ba
&\int\limits_{\supp\rho}|\ln\|(x,x)\|_{\rz^2}-\ln\|(y'-x,y'-x)\|_{\rz^2}|^2\ud x\\
&\hspace{0.25cm}=\int\limits_{\{(x,x)\notin B_{2\delta(\epsilon)}(0,0)\}\cap\supp\rho}|\ln\|(x,x)\|_{\rz^2}-\ln\|(y'-x,y'-x)\|_{\rz^2}|^2\ud x\\
&\hspace{0.5cm}+\int\limits_{\{(x,x)\in B_{2\delta(\epsilon)}(0,0)\}\cap\supp\rho}|\ln\|(x,x)\|_{\rz^2}-\ln\|(y'-x,y'-x)\|_{\rz^2}|^2\ud x\\
&\hspace{0.5cm}\leq C\epsilon+C'\int\limits_{0}^{2\delta(\epsilon)}|\ln|1-\frac{y'}{x}||^2\ud x\\,
&\hspace{0.5cm}\leq C\epsilon+C'\int\limits_{y'(2\delta(\epsilon))^{-1}}^{\infty}\frac{y'}{x^2}|\ln|1-x||^2\ud x=C\epsilon+\Or(y'),
\ea
\ee
where in the last line we performed the substitution $x\rightarrow\frac{y}{x}$, and we used \cite[pp.~240,241]{MR874986}. 
Observing that $C$,~$C'>0$ only depend on the size of $|\supp\rho|_1$ and choosing $y'$ sufficiently close $0$ proves the claim.

In the same way as above, we view $(\sqrt{|\rho|}\mf{b}(\overline{k}))^{\ast}$ as a continuous map from $(\rz^2,\|\cdot\|_{\rz^2})$ to $L^2(\rz)$. This gives the same integral kernel $\tilde{b}$ as in \eqref{tildeb} since $-\ui k\in\rz^+$ and then the $K_0$-Bessel function is real valued, but the $\rz^2$-variables are now indicated by $(x_1,x_2)$. We denote for convenience, see \eqref{betterres}, 
\be
q=(\eins+\sgn\rho\sqrt{|\rho|}\mf{g}(k)\sqrt{|\rho|})^{-1}:L^2(\rz)\rightarrow L^2(\rz)\,,
\ee
and since $q$ is continuous we observe that the integral kernel $r_{\rho}(k)((x_1,x_2),(y_1,y_2))$ of $R^{\reg}_{\rho}(k)$ is given by
\be
\label{rre}
r_{\rho}(k)((x_1,x_2),(y_1,y_2))=\langle\tilde{b}(x_1,x_2),q\tilde{b}(y_1,y_2)\rangle_{L^2(\rz)}\,.
\ee
Now, the algebraic identity
\be
\label{inequ}
\ba
&r_{\rho}(k)((x_1,x_2),(y_1,y_2))-r_{\rho}(k)((x_1',x_2'),(y_1',y_2)')\\
&\hspace{0.25cm}=\langle\tilde{b}(x_1,x_2),q\tilde{b}(y_1,y_2)\rangle_{L^2(\rz)}-\langle\tilde{b}(x_1',x_2'),q\tilde{b}(y_1',y_2')\rangle_{L^2(\rz)}\\
&\hspace{0.25cm}=\langle(\tilde{b}(x_1,x_2)-\tilde{b}(x_1',x_2')),q\tilde{b}(y_1,y_2)\rangle_{L^2(\rz)}+\langle\tilde{b}(x_1',x_2'),q(\tilde{b}(y_1,y_2)-\tilde{b}(y_1',y_2'))\rangle_{L^2(\rz)}\,,
\ea
\ee
and a suitable application of H\"older's inequality prove the first part of the claim.

To see the exponential decay we recall that the operators $\tilde{b}$ in \eqref{rre} involve the kernel $\sqrt{|\rho(x)|}b(k)(x,y_1,y_2)$ given in \eqref{b}. As before, it's enough to assume $x\in\supp{\rho}$, and since the support of $\rho$ is finite we have for sufficiently large $\|(y_1,y_2)\|_{\rz^2}$ the inequality
\be
\label{ineq2a}
\|(y_1-x,y_2-x)\|_{\rz^2}>\frac{1}{2}\|(y_1,y_2)\|_{\rz^2}\,.
\ee
We obtain, using H\"older's inequality,
\be
\label{largexy}
\ba
&|r_{\rho}(k)((x_1,x_2),(y_1,y_2))|=|\langle\tilde{b}(x_1,x_2),q\tilde{b}(y_1,y_2)\rangle_{L^2(\rz)}|\\
&\hspace{0.5cm}\leq\|\tilde{b}(x_1,x_2)\|_{L^2(\rz)}\|q\|\|\tilde{b}(y_1,y_2)\|_{L^2(\rz)}\,,
\ea
\ee
where $\|q\|$ is the $L^2$-operator norm of $q$. For sufficiently large $\|(y_1,y_2)\|_{\rz^2}$ we may use \eqref{largez} for \eqref{b}. Together with \eqref{ineq2a} we then obtain 
\be
\label{l2normlarge}
\ba
&\|\tilde{b}(y_1,y_2)\|_{L^2(\rz)}\leq C\int\limits_{\supp \rho}\ue^{-k\sqrt{(x-y_1)^2+(x-y_2)^2}}\ud x\\
&\hspace{0.25cm}\leq C'\ue^{-k\frac{1}{2}\sqrt{y_1^2+y_2^2}}
\ea
\ee
with some $C$,~$C'>0$. Plugging \eqref{l2normlarge} in \eqref{largexy} proves the second part of the claim.
\end{proof} 
Having established the existence of a (continuous) integral kernel we take over the notation of the above proof. We denote for $k\in C_{\alpha}$, $\alpha<\frac{\pi}{2}$, and $|k|$ sufficiently large, the integral kernel of $R^{\reg}_{\rho}(k)$ by $r_{\rho}(k)(\cdot,\cdot)$. 

The following proposition tells us how we may calculate the trace of the regularized resolvent.
\begin{prop}
\label{propregr}
The trace of the regularized resolvent may be calculated by, $k\in C_{\alpha}$, $\alpha<\frac{\pi}{2}$, and $|k|$ sufficiently large,
\be
\tr R^{\reg}_{\rho}(k)=\int\limits_{\rz^2}r_{\reg}^{\rho}(k)(\bs{x},\bs{x})\ud \bs{x}\,.
\ee
\end{prop}
\begin{proof}
The claim follows by an application of \cite[p.~117]{Gohberg:1969} saying that the properties of $R^{\reg}_{\rho}(k)$ and of its kernel $r_{\reg}^{\rho}(k)$ derived in Proposition~\ref{restraceclass} and Lemma~\ref{lemma1} are sufficient to deduce the claim. 
\end{proof}
The above (trace-class) result tempts us to define the regularized trace of the resolvent as the trace of the regularized resolvent
\begin{defn}
\label{defregr}
The regularized trace of the resolvent $R_{\rho}(\lambda)$ is defined as
\be
\label{trregr}
\tr_{\reg}R_{\rho}(\lambda):=\tr R^{\reg}_{\rho}(\sqrt{\lambda})\,, 
\ee
with $\lambda$ such that $k=\sqrt{\lambda}$ satisfies the assumption of Proposition~\ref{propregr}.
\end{defn}

\subsection{An asymptotic analysis of the trace of the regularized resolvent} To determine the asymptotic expansion of the trace of the regularized resolvent we have to introduce a couple of auxiliary objects and notations permitting a closed presentation. 

We start with defining a diffeomorphism, $n\in\nz_0$, 
\be
\label{phiy}
\phi:\rz^{n+1}\rightarrow \rz^{n+1}\,,
\ee
by 
\be
\label{trans1}
\ba
w_{l}&=\phi(y_0,\ldots,y_n)_{l}=y_{l}-y_{l+1}\,, \quad l\in\{0,\ldots,n-1\}\,,\\
w_{n}&=\phi(y_0,\ldots,y_n)_{n}=y_{n}+y_0\,.
\ea
\ee
Note that for $n=0$ we have $w_0=2y_0$. The inverse map (diffeomorphism) 
\be
\label{phi-1}
\phi^{-1}:\rz^{n+1}\rightarrow \rz^{n+1}
\ee
reads as
\be
y_l=\phi^{-1}(w_0,\ldots,w_n)_l=\frac{1}{2}[\sum\limits_{m=l}^n w_m-\sum\limits_{l=0}^{l-1}w_m],\quad l\in\{0,\ldots,n\}\,,
\ee
and, in particular, for $n=0$ we have $y_0=\frac{1}{2}w_0$. Moreover, we are going to utilize the following (combinatorial) set of maps, $n\in\nz_0$,
\be
\label{S_n}
S_n:=\begin{cases}
\{s: \quad s:\{0,\ldots,n-1\}\rightarrow\{1,-1\}\}\,, & n\in\nz\,,\\
\emptyset\,, & n=0\,.
\end{cases}
\ee
In addition, we employ the multy-index notation
\be
\label{alphan}
\alpha^n:=(\alpha_0,\ldots,\alpha_{n-1})\,,\quad \alpha_l\in \nz_0, \quad l\in\{0,\ldots,n-1\}\,.
\ee
together with
\be
\label{alphanorm}
|\alpha^n|:=\sum\limits_{l=0}^{n-1}\alpha_l\,.
\ee
We remark here that $\alpha^n$ is only defined for $n\neq0$.

For any $\rho\in C^{\infty}_0(\rz)$ the maps $\phi^{-1}$ in \eqref{phi-1} and $s\in S_n$, $n\in\nz_0$, in \eqref{S_n} are employed to generate a smooth and compactly supported map $\rho_{n,s}$ from $\rz^{n+1}$ to $\rz$, i.e., $\rho_{n,s}\in C^{\infty}_0(\rz^{n+1},\rz)$, defined by, $(w_0,\ldots,w_n)\in\rz^{n+1}$,
\be
\label{phins}
\rho_{n,s}(w_0,\ldots,w_n):=
\begin{cases}
\prod\limits_{l=0}^{n}\rho((\phi^{-1}(s(0)w_0,\ldots,s(n-1)w_{n-1},w_n))_l)\,, & \ n\in\nz\,,\\
\rho(\frac{w_0}{2})\,, & n=0\,.
\end{cases}
\ee
Our asymptotic analysis also deploys the following notations of partial derivatives of $\rho_{n,s}$, using \eqref{phins} and \eqref{alphan},
\be
\label{partial2}
{\partial}_{\alpha^n}\rho_{n,s}(w_0,\ldots,w_{n-1},w_n):=(\frac{\partial^{|\alpha^n|}}{\partial^{\alpha_0}_{w_0}\ldots\partial^{\alpha_{n-1}}_{w_{n-1}}}\rho_{n,s})(w_0,\ldots,w_{n-1},w_n)\,.
\ee
Finally, fixing $\rho\in C^{\infty}_{0}(\rz)$, the following functions will turn out of particular interest
\be
\label{simple1}
\ba
&c_{\alpha^n,s,l}\\
&\hspace{0.25cm}=
\begin{cases}
\int\limits_{{\rz^+}^{n}}\int\limits_{\rz}\frac{\ud \xi\ud t_0\ldots\ud t_{n-1}}{(1+\xi^2)^{\frac{3}{2}}(\sqrt{2}((\cosh(t_{0})+\ui\xi s(n-1)))^{\alpha^n_0+1}\ldots(\sqrt{2}(\cosh(t_{n-1})+\ui\xi s(n-1)))^{\alpha^n_{n-1}+1}}\,,&n\in\nz\,,\\
\int\limits_{\rz}\frac{\ud \xi}{(1+\xi^2)^{\frac{3}{2}}}\,, &n=0\,, l=0\,,\\
0\,,&\text{else}\,,
\end{cases}
\ea
\ee
and, $l\in\nz_0$,
\be
\label{simple2}
b_{n,l}:=
\begin{cases}
\sum\limits_{|\alpha^n|=l,\atop s\in S_n}c_{\alpha^n,s,l}\int\limits_{\rz}\partial_{\alpha^n}\rho_{n,s}(0,\ldots,0,y)\ud y\,,& n\in\nz\,,\\
\int\limits_{\rz}2\rho(\frac{y}{2})\ud y\,, & n=0, l=0\,,\\
0,&\text{else}\,,
\end{cases}
\ee
where we incorporated \cite[3.251~11.]{MR2360010} and $S_0=\emptyset$, \eqref{S_n}. We remark that by \cite[3.252~11.]{MR2360010}
\be
\label{c0}
c_{\alpha^0,s,0}=2\,.
\ee

Equipped with the above identities, we are now able to determine the large-$\lambda$ asymptotic expansion of the regularized trace of the resolvent. For this, we remind that $C_{\alpha}$ is a sector with opening angle $\alpha$ around the positive imaginary axis, \eqref{calpha}. 
\begin{theorem}
\label{theor}
The regularized trace of the resolvent $\tr_{\reg} R_{\rho}(\lambda)$ possesses for $|\lambda|\rightarrow\infty$ and $k:=\sqrt{-\lambda}\in C_{\alpha}$ with $\alpha<\frac{\pi}{2}$ a complete asymptotic expansion in integer powers of $k$ of the form
\be
\label{asumpr}
\tr_{\reg} R_{\rho}(-\lambda)\sim\sum\limits_{m=0}^{\infty}b_m\lambda^{-(\frac{m}{2}+1)}\,,
\ee
where the coefficients $b_m$ are given by
\be
\label{bmk}
b_m=\frac{1}{8}\sum\limits_{n,l,\atop l+n=m}(-2\pi)^{-(n+1)}b_{n,l}\,.
\ee
The first two coefficients read as
\be
\label{b01}
\ba
b_0&=-\frac{1}{4\pi}\int\limits_{\rz}\rho(y)\ud x\,,\quad b_1=\frac{\sqrt{2}}{32}\int\limits_{\rz}\rho(y)^2\ud y\,.
\ea
\ee
\end{theorem}
\begin{rem}
The condition on $\lambda$ implies that $\lambda$ has to be in a cone around the positive axis $\rz^+$ with opening angle smaller than $2\pi$. Moreover, we point out that for $m\geq2$ integrals of derivatives of $\rho$ appear in \eqref{bmk} for $b_m$. 
\end{rem}
\begin{proof}
For our convenience we choose $k=\sqrt{-\lambda}$ and consider only the case $\tilde{k}:=-\ui k\in\rz^+$, $|k|$ sufficiently large. The general case $k\in C_{\alpha}$ may be treated analogously. We also remark that in the following every interchange of the order of integration is justified by Fubini's theorem, \cite[23.7 Corollary]{Bauer:2001}.

We are going to use the resolvent representation \eqref{resolvent}. First, we expand $(\eins+\rho\mf{g}(k))^{-1}$ into a Neumann series and attain
\be
\label{neumann}
\ba
\tr_{\reg} R_{\rho}(-\lambda)&=-\tr(\mf{b}\left(\overline{k}\right)^{\ast}(\eins+\rho\mf{g}(k))^{-1}\rho \mf{b}(k))\\
&\hspace{0.25cm}=\sum\limits_{n=0}^{\infty}(-1)^{n+1}\tr(\mf{b}\left(\overline{k}\right)^{\ast}(\rho\mf{g}(k))^n\rho \mf{b}(k))\,.
\ea
\ee
It is possible to get a large-$k$ asymptotic expansion of $\tr(\mf{b}\left(\overline{k}\right)^{\ast}(\rho\mf{g}(k))^n\rho \mf{b}(k))$ for every $n\in\nz_0$, and then we rearrange the terms w.r.t. powers of $\tilde{k}$ in \eqref{neumann}. To see the first part of the afore mentioned, we use the integral kernels \eqref{gk} and \eqref{b} for $\mf{b}(k)$ and $\mf{g}(k)$, and we use the notation $Y=(y_0,y_1,\dots,y_n)\in\rz^{n+1}$, giving 
\be
\label{inttrace2}
\ba
&\tr(\mf{b}\left(\overline{k}\right)^{\ast}(\rho\mf{g}(k))^n\rho \mf{b}(k))\\
&\hspace{0.25cm}=(2\pi)^{-(n+2)}\int\limits_{\rz^2}\int\limits_{\rz^{n+1}}K_0(\tilde{k}\sqrt{(x_1-y_0)^2+(x_2-y_0)^2})\rho(y_0)\ast\ldots\\
&\hspace{0.5cm}\ldots\ast K_0(\sqrt{2}\tilde{k}|y_0-y_1|)\rho(y_{n-1})K_0(\sqrt{2}\tilde{k}|y_{n-1}-y_n|)\ast\ldots\\
&\hspace{0.5cm}\ldots\ast\rho(y_n)K_0(\tilde{k}\sqrt{(y_n-x_1)^2+(y_n-x_2)^2})\ud Y\ud x_1\ud x_2\,.
\ea
\ee
Now, by slight abuse of notation, we apply the (orthogonal) coordinate transformation 
\be
(x_1,x_2)\rightarrow\frac{1}{\sqrt{2}}(x_1+x_2,x_1-x_2)\,,
\ee
followed by an insertion of \eqref{integralrep1} in \eqref{inttrace2}, using the notation $T=(t_0,\dots,t_{n-1})\in {\rz^+}^n$, and \eqref{intKK1} in \eqref{inttrace2} which yields
\be
\label{inttrace}
\ba
&\tr(\mf{b}\left(\overline{k}\right)^{\ast}(\rho\mf{g}(k))^n\rho \mf{b}(k))\\
&\hspace{0.25cm}=(2\pi)^{-(n+2)}\int\limits_{\rz^{n}}\int\limits_{\rz^{n+1}}\int\limits_{\rz^2}K_0(\tilde{k}\sqrt{(x_1-\sqrt{2}y_0)^2+x_2^2})\rho(y_0)\ast\ldots\\
&\hspace{0.5cm}\ldots\ast\ue^{-\sqrt{2}\tilde{k}|y_0-y_1|\cosh(t)}\rho(y_{n-1})\ue^{-\sqrt{2}\tilde{k}|y_{n-1}-y_n|\cosh(t)}\ast\ldots\\
&\hspace{0.5cm}\ldots\ast\rho(y_{n})K_0(\tilde{k}\sqrt{(\sqrt{2}y_n-x_1)^2+x_2^2})\ud x_1 \ud x_2\ud Y\ud T\\
&\hspace{0.25cm}=\frac{\pi}{2 \tilde{k}^2}(2\pi)^{-(n+2)}\int\limits_{{\rz^+}^{n}}\int\limits_{\rz}\int\limits_{\rz^{n+1}}\frac{\ue^{-\ui\sqrt{2} \tilde{k} \xi(y_0-y_n)}}{(1+\xi^2)^{\frac{3}{2}}}\rho(y_0)\ue^{-\sqrt{2}\tilde{k}|y_0-y_1|\cosh(t_0)}\ast\ldots\\
&\hspace{0.5cm}\ldots\ast\rho(y_{n-1})\ue^{-\sqrt{2}\tilde{k}|y_{n-1}-y_n|\cosh(t_{n-1})}\rho(y_n)\ud Y\ud \xi\ud T\,.
\ea
\ee 
For $n=0$ we directly calculate the trace and obtain
\be
\label{n=0}
\ba
&\tr(\mf{b}\left(\overline{k}\right)^{\ast}\rho \mf{b}(k))=\frac{1}{8\pi \tilde{k}^2}\int\limits_{\rz}\int\limits_{\rz}\frac{1}{(1+\xi^2)^{\frac{3}{2}}}\rho(y_0)\ud \xi\ud y_0\\
&\hspace{0.25cm}=\frac{1}{4\pi \tilde{k}^2}\int\limits_{\rz}\rho(y_0)\ud y_0\,,
\ea
\ee
where in the last line we used \cite[3.252~11.]{MR2360010}. For $n\neq0$ we want to invoke for our asymptotic analysis the integration by parts method. For this, it is expedient to first use \eqref{phi-1} as an appropriate substitution of variables. This transformation implies
\be
\label{trans2}
\sum\limits_{l=0}^{n-1}w_l=y_0-y_n\,,
\ee
and the determinant of the Jacobian $\det J(\phi^{-1})$ of this coordinate transformation is constant, and given by $\det J(\phi^{-1}) = \frac{1}{2}$. Hence, changing the variables in \eqref{inttrace}, using \eqref{trans1}, \eqref{trans2} and $W=(w_0,\ldots,w_{n-1})\in\rz^{n}$, gives
\be
\label{inttrace1}
\ba
&\int\limits_{{\rz^+}^{n}}\int\limits_{\rz}\int\limits_{\rz^{n+1}}\frac{\ue^{-\ui \tilde{k} \xi(y_0-y_n)}}{(1+\xi^2)^{\frac{3}{2}}}\rho(y_0)\ue^{-\sqrt{2}\tilde{k}|y_0-y_1|\cosh(t_0)}\ast\ldots\\
&\hspace{0.5cm}\ldots\ast\rho(y_{n-1})\ue^{-\sqrt{2}\tilde{k}|y_{n-1}-y_n|\cosh(t_{n-1})}\rho(y_n)\ud Y\ud \xi\ud T\\
&\hspace{0.25cm}=\frac{1}{2}\int\limits_{{\rz^+}^{n}}\int\limits_{\rz}\int\limits_{\rz}\int\limits_{\rz^{n+1}}\frac{1}{(1+\xi^2)^{\frac{3}{2}}}\rho(\phi^{-1}(w_0,\ldots,w_n)_0)\ue^{-\sqrt{2}\tilde{k}(|w_0|\cosh(t_0)+\ui\xi w_0)}\ast\ldots\\
&\hspace{0.5cm}\ldots\ast\rho(\phi^{-1}(w_0,\ldots,w_n)_{n-1})\ue^{-\sqrt{2}\tilde{k}(|w_{n-1}|\cosh(t_{n-1})+\ui\xi w_{n-1})}\ast\\
&\hspace{0.5cm}\ast\rho(\phi^{-1}(w_0,\ldots,w_n)_n)\ud W\ud w_{n}\ud\xi\ud T\\
&\hspace{0.25cm}=\sum\limits_{s\in S_n}\frac{1}{2}\int\limits_{{\rz^+}^{n}}\int\limits_{\rz}\int\limits_{\rz}\int\limits_{{\rz}^{n}}\frac{1}{(1+\xi^2)^{\frac{3}{2}}}\ast\ldots\\
&\hspace{0.5cm}\ldots\ast\ue^{-\sqrt{2}\tilde{k}(\sgn(w_0)\cosh(t_0)+\ui\xi)w_0}\ldots\ue^{-\sqrt{2}\tilde{k}(\sgn(w_{n-1})\cosh(t_{n-1})+\ui\xi)w_{n-1}}\ast\\
&\hspace{0.5cm}\ast\prod\limits_{l=0}^{n}\rho((\phi^{-1}(w_0,\ldots,w_{n-1},w_n))_l)\ud W\ud w_n\ud\xi\ud T\\
&\hspace{0.25cm}=\sum\limits_{s\in S_n}\frac{1}{2}\int\limits_{{\rz^+}^{n}}\int\limits_{\rz}\int\limits_{\rz}\int\limits_{{\rz^+}^{n}}\frac{1}{(1+\xi^2)^{\frac{3}{2}}}\rho_{n,s}(w_0,\ldots,w_n)\ast\\
&\hspace{0.5cm}\ast\ue^{-\sqrt{2}\tilde{k}(\cosh(t_0)+\ui\xi s(0))w_0}\ldots\ue^{-\sqrt{2}\tilde{k}(\cosh(t_{n-1})+\ui\xi s(n-1))w_{n-1}}\ud W\ud w_n\ud\xi\ud T\,.
\ea
\ee

It is for the following integration by parts method important that in the last line of \eqref{inttrace1} only the variables $w_l$ with $l\in\{0,\ldots,n-1\}$ appear in the exponential function. We perform an integration by parts w.r.t. the $W$ variables. The obtained terms which don't possess any $W$ integrals anymore may then be ordered w.r.t. powers of $\tilde{k}$. Using our notation \eqref{alphanorm} and \eqref{partial2}, we obtain, $l\in\nz_0$,
\be
\label{byparts}
\ba
&\int\limits_{{\rz^+}^{n}}\int\limits_{\rz}\int\limits_{\rz}\int\limits_{{\rz^+}^{n}}\frac{1}{(1+\xi^2)^{\frac{3}{2}}}\rho_{n,s}(w_0,\ldots,w_n)\ast\\
&\hspace{0.5cm}\ast\ue^{-\sqrt{2}\tilde{k}(\cosh(t_0)+\ui\xi \rho(0))w_0}\ldots\ue^{-\sqrt{2}\tilde{k}(\cosh(t_{n-1})+\ui\xi s(n-1))w_{n-1}}\ud W\ud w_n\ud\xi\ud T\\
&\hspace{0.25cm}=\sum\limits_{\alpha^n,\atop|\alpha^n|\leq l}\tilde{k}^{-|\alpha^n|-n}\int\limits_{{\rz^+}^{n}}\int\limits_{\rz}\int\limits_{\rz^+}\frac{1}{(1+\xi^2)^{\frac{3}{2}}}{\partial}_{\alpha^n}\rho_{n,s}(0,\ldots,0,w_n)\ast\\
&\hspace{0.5cm}\ast\frac{\ud w_n\ud \xi\ud T}{(\sqrt{2}(\cosh(t_{0})+\ui\xi s(n-1)))^{\alpha^n_0+1}\ldots(\sqrt{2}(\cosh(t_{n-1})+\ui\xi s(n-1)))^{\alpha^n_{n-1}+1}}\\
&\hspace{0.5cm}+\Or(\tilde{k}^{-(l+n+1)})\,.
\ea
\ee 
The last line follows by integration by parts and the observation that every partial derivative evaluated at $(0,\ldots,0,w_n)$ is generated exactly once. Moreover, we used that the $T$ and $\xi$ integration don't affect the order estimate $\Or(\tilde{k}^{-(l-n+1)})$.

Now, sorting \eqref{byparts} w.r.t. powers of $\tilde{k}$ by making use of our definitions \eqref{simple1} and \eqref{simple2} we get
\be
\label{simple3}
\ba
&\tr(\mf{b}\left(\overline{k}\right)^{\ast}(\rho\mf{g}(k))^n\rho \mf{b}(k))=\frac{1}{8\tilde{k}^2}\frac{1}{(2\pi)^{n+1}}\sum\limits_{l'=0}^{l}\frac{1}{\tilde{k}^{l'+n}}b_{n,l'}+O(\tilde{k}^{-(l+n+1)})\,.
\ea
\ee
Note that \eqref{simple3} is conform with \eqref{n=0} for $l=n=0$.

Finally, we use \eqref{neumann} plugging in there the asymptotic expansion \eqref{simple3}, and we sort the obtained sum again w.r.t. powers of $\tilde{k}$. This then yields the asymptotic expansion w.r.t. powers of $\ui\tilde{k}=k=\sqrt{-\lambda}$, \eqref{bmk}. The calculations of the first two coefficients are as follows
\be
\label{bs}
b_0=-\frac{1}{16\pi}b_{0,0}\,,\quad b_1=-\frac{1}{16\pi}b_{0,1}+\frac{1}{32\pi^2}b_{1,0}\,,
\ee  
and it remains to calculate the three coefficients in \eqref{bs} by \eqref{simple2}. The first two ones are simple and given by, \eqref{simple2},
\be
\label{b00}
b_{0,0}=4\int\limits_{\rz}\rho(y)\ud y\,, \quad b_{0,1}=0\,, 
\ee
where we used \eqref{c0}. For $b_{1,0}$ we get by \eqref{simple2} 
\be
\label{b10}
b_{1,0}=c_{\alpha^n,s_+,0}\int\limits_{\rz}\rho_{0,s_+}(0,y)\ud y+c_{\alpha^n,s_-,0}\int\limits_{\rz}\rho_{0,s_-}(0,y)\ud y\,,
\ee
where $s_+(0):=1$ and $s_-(0):=-1$. We have for $|\alpha^1|=0$ the simple relations
\be
\label{sigmas0}
\rho_{0,s_\pm}(0,y)=(\rho(\frac{1}{2}y))^2\,,
\ee
and, $|\alpha^1|=0$,
\be
\label{cn=1s}
c_{\alpha^1,s_\pm,0}=\frac{1}{\sqrt{2}}\int\limits_{\rz^+}\int\limits_{\rz}\frac{1}{(\xi^2+1)^{\frac{3}{2}}}\frac{\cosh(t_0)\mp\ui \xi}{\cosh(t_0)^2+\xi^2}\ud t_0\ud\xi\,.
\ee
Taking into account that the imaginary part cancels out in \eqref{cn=1s} we are only interested on the real part of \eqref{cn=1s} given by, \cite[2.5.49.~3.]{MR874986},
\be
\label{cn+1sreal}
c_{\alpha^1,s_\pm,0}=\re c_{\alpha^1,s_\pm,0}=\frac{\pi}{\sqrt{2}}\int\limits_{\rz}\frac{1}{(1+\xi^2)^2}\ud \xi=\frac{\pi^2}{2\sqrt{2}}\,.
\ee 
Hence, inserting \eqref{cn+1sreal} and \eqref{sigmas0} into \eqref{b10} gives
\be
\label{b10a}
b_{1,0}=\frac{\pi^2}{\sqrt{2}}\int\limits_{\rz}\rho(\frac{1}{2}x)^2\ud x=\sqrt{2}\pi^2\int\limits_{\rz}\rho(x)^2\ud x\,.
\ee
Now, plugging \eqref{b10a} and \eqref{b00} into \eqref{bs} proves the claim.
\end{proof}
Regarding Theorem~\ref{theor}, we make the following remark.
\begin{remark}
On the r.h.s. of formula~\eqref{simple1} only the real part is essential as the imaginary part vanishes.
\end{remark}
\section{The asymptotic expansion of the regularized trace of the heat kernel}
Armed with all the results inferred in our paper so far we are now in the position to deduce the existence of the heat kernel and to conclude the small-$t$ asymptotic expansion of the regularized trace of the heat semi-group $\ue^{t\Delta_{\rho}}$ (and heat kernel). 

We are going to exploit that the resolvent of a contraction semi-group admits a representation as a Laplace transformation of the heat semi-group, \cite[Satz~VII.4.10]{MR1787146}. As for the resolvent, the heat semi-group $\ue^{t\Delta_{\rho}}$ isn't a trace-class operator due to the presence of an essential spectrum of $-\Delta_{\rho}$. Again, we may regularize the trace of the heat semi-group analogously to Definition~\ref{defregr} by subtracting the free heat semi-group, i.e., $\{\ue^{t\Delta_{\rho}}\}_{\reg}:=\ue^{t\Delta_{\rho}}-\ue^{t\Delta_{0}}$, $t>0$.

To see that $\{\ue^{t\Delta_{\rho}}\}_{\reg}$ is trace class as well, and  how we may calculate its trace, we prove the following lemma.
\begin{lemma}
\label{lemmatraceheat}
The operator $\{\ue^{t\Delta_{\rho}}\}_{\reg}$ is for $t>0$ a trace-class integral operator with kernel $k_{\reg}^{\rho}(t)(\cdot,\cdot)\in C^{\infty}(\rz^2\times\rz^2)\cap L^{\infty}(\rz^2\times\rz^2)$. Its trace is given by
\be
\label{tracek}
\tr\{\ue^{t\Delta_{\rho}}\}_{\reg}=\int\limits_{\rz^2}k_{\reg}^{\rho}(t)(\bs{x},\bs{x}) \ud \bs{x}\,.
\ee
\end{lemma}
\begin{proof}
With the same Dunford-Pettis argument as in \cite[Lemma~6.1]{MR2277618} we may infer that $\{\ue^{t\Delta_{\rho}}\}_{\reg}$ is an integral operator possessing a smooth and bounded kernel for $t>0$. We shall use the Dunford-Taylor integral identity, \cite[Section~IX.1.6]{Kato:1966},
\be
\{\ue^{t\Delta_{\rho}}\}_{\reg}=\frac{\ui}{2\pi}\int\limits_{\gamma}\ue^{-\lambda t}R^{\reg}_{\rho}(\sqrt{\lambda})\ud \lambda\,,
\ee
where $\gamma$ is a suitable contour encircling the spectrum of $-\Delta_{\rho}$ in a positively orientated way. With a similar method as in the proof of Lemma~\ref{lemma1} we may infer that $R^{\reg}_{\rho}(\sqrt{\lambda})$ is continuous in trace norm for suitable $\gamma$'s. Furthermore, due to the asymptotics \eqref{asumpr} we conclude that the integral converges in trace norm. Now, \cite[Satz~3.22]{Weidmann:2000} proves the first part of the claim. The second part my be proven analogously to Proposition \ref{propregr} incorporating the above properties of the integral kernel $k_{\reg}^{\rho}(t)$. 
\end{proof}
It is reasonable to define the regularized trace of the heat semi-group (and heat kernel) analogously to \eqref{trregr}. 
\begin{defn}
\label{regheat}
The regularized trace of the heat semi-group (heat kernel) is defined as
\be
\tr_{\reg}\ue^{t\Delta_{\rho}}:=\tr\{\ue^{t\Delta_{\rho}}\}_{\reg}=\int\limits_{\rz^2} k_{\reg}^{\rho}(t)(\bs{x},\bs{x})\ud\bs{x}\,, \quad t>0\,.
\ee 
\end{defn}
We are now ready to present the result concerning our desired small-$t$ asymptotic expansion of the regularized trace of the heat semi-group (heat kernel).
\begin{theorem}
\label{theorheat}
Let $\rho\in C^{\infty}_0(\rz)$. Then, the regularized trace of the heat semi-group resp. heat kernel possesses a complete asymptotic expansion in powers of $t$ given by, $t\rightarrow 0$,
\be
\tr_{\reg}\ue^{t\Delta_{\rho}}\sim\sum\limits_{n=0}^{\infty}a_nt^{\frac{n}{2}}\,,
\ee
where 
\be
\label{as}
\ba
a_{2n}&=\frac{b_{2n}}{n!}\,, \ n\in\nz_0\,, \quad a_{2n+1}=\frac{n! 2^{2n+1}b_{2n+1}}{(2n+1)!\sqrt{\pi}}\,, \ n\in\nz_0\,,
\ea
\ee
and the $b_n$'s are given in \eqref{bmk}.
\end{theorem}
\begin{proof}
In view of Lemma~\ref{lemmatraceheat}, we may utilize the well-known identity, \cite[Satz~VII.4.10]{MR1787146},
\be
\tr_{\reg}R_{\rho}(-\lambda)=\int\limits_{\rz^+}\ue^{-\lambda t}\tr_{\reg}\ue^{t\Delta_{\rho}}\ud t\,,
\ee
with $\re\lambda>0$ and $\lambda>|\lambda_{\min,\rho}|$ (sufficiently large). Now, we want to apply the converse Watson lemma, Lemma~\ref{lemmaw}. For this, we observe that the condition $|\arg(|\lambda_{\min,\sigma}|-\lambda|)\leq\frac{\pi}{2}$ is equivalent to $k\in C_{\alpha}$ with $\alpha\leq\frac{\pi}{4}$. Hence, we may apply Lemma~\ref{lemmatraceheat} and use the asymptotic expansion of $\tr_{\reg}R_{\sigma}(-\lambda)$ in Theorem~\ref{theor}. Comparing \eqref{asumpr} with \eqref{iwl} we infer that $\lambda_{n}=\frac{n}{2}+1$. Finally, we use \cite[pp.~2,3]{Magnus:1966} to calculate 
\be
\label{gamman12}
\Gamma(n+\frac{3}{2})=(n+\frac{1}{2})\Gamma(n+\frac{1}{2})=\frac{(2n+1)\sqrt{\pi}}{n! 2^{2n+1}}\,,
\ee
and we plug \eqref{gamman12} together with the $b_n$'s, \eqref{bmk}, in \eqref{iwl}. That reveals the identity \eqref{as}.
\end{proof}
We end this paper with a comparison of our heat kernel asymptotics with known results for a Schr\"odinger operator on $\rz^2$, however, with a smooth potential, $V$. Using the Theorems~\ref{theorheat}~and~\ref{theor} we obtain for our system the leading asymptotic estimate
\be
\tr_{\reg}\ue^{t\Delta_{\sigma}}=-\frac{1}{4\pi}\int\limits_{\rz}\sigma(x)\ud x+\frac{\sqrt{2}\sqrt{t}}{16\sqrt{\pi}}\int\limits_{\rz}\sigma^2(x)\ud x+\Or(t)\,, \quad t\rightarrow0\,.
\ee 
On the other hand, for a Schr\"odinger operator of the form $-\Delta+V(\cdot)$ with $V\in C^\infty_0(\rz^2)$ the result on \cite[p,~405]{MR1994690} is (using an analogous notation), $t\rightarrow0^+$,
\be
\tr_{\reg}\ue^{t(\Delta-V(\cdot))}=-\frac{1}{4\pi}\int\limits_{\rz^2}V(\bs{x})\ud \bs{x}+\frac{{t}}{24\pi}\int\limits_{\rz^2}(3V^2(\bs{x})-\Delta V(\bs{x}))\ud\bs{x}\,.
\ee
We see that the first coefficients and the power of $t$ agree, but not the second coefficient and the corresponding power of $t$ owing to the fact that our potential is supported only on a codimension one submanifold. 
\subsection*{Acknowledgment}
The author is very grateful to Ram Band for helpful discussions and comments and he is indebted to Frank Steiner for pointing out various useful relations of Bessel functions. The work has been supported by ISF (Grant~No.~494/14).
{\small
\bibliographystyle{abbrv}
\bibliography{Literature}
}
\appendix
\section{Notations}
\label{not}
First, we introduce some notations and denote by $\langle\cdot,\cdot\rangle_{\rz^2}$ and $\|\cdot\|_{\rz^2}$ the standard inner product and the standard Euclidean norm on $\rz^2$, respectively. To ease notation we denote points in $\rz^2$ by bold letters, e.g., $\bs{x}=(x_1,x_2)$. We put $0\leq\arg z<2\pi$ as the range of the argument for $z\in\kz$, and set the brunch cut of the square root $\sqrt{\cdot}$ on $\rz^+$ such that $\sqrt{z}=\ui\sqrt{|z|}$ for $-z\in\rz^+$. By $|\cdot|_1$ we denote the one-dimensional Hausdorff measure. The spectrum of an operator $O$ is denoted by $\sigma(O)$ and by $\|O\|$ we refer to the standard operator norm, \cite[Sections~2.1,~5.1]{Weidmann:2000}. Moreover, we use standard notations for the set of $n$-times continuously differentiable functions (possessing compact support), $C^{n}(\rz^{m})$ ($C^{n}_0(\rz^{m})$), on $\rz^m$, and for the set of square Lebesgue-integrable functions, $L^2(\rz^m)$, on $\rz^{m}$.

We put the Fourier transform $Ff$ of a suitable function $f$ on $\rz$ as
\be
(Ff)(\xi):=\frac{1}{\sqrt{2\pi}}\int\limits_{\rz}\ue^{-\ui\xi x}f(x)\ud x\,.
\ee
We also remind that by Plancherel's theorem the Fourier transform generates a unitary map on $L^2(\rz)$, \cite[Section~2.2.3]{Grafakos:classical}. In addition, by $\Gamma(\cdot)$ we identify the Gamma function and by $K_{0}(\cdot)$ the $K_0$-Bessel function (Mcdonald function)), \cite[pp.~1,66]{Magnus:1966}.
\section{Integral identities for the $K_0$-Macdonald function}
First, we remind the asymptotic behavior of the $K_0$-Macdonald function for large and small arguments, $\gamma$ Euler-Mascheroni constant, \cite[p.~69]{Magnus:1966}, 
\be
\label{smallz}
K_0(z)=-(\ln(\frac{z}{2})+\gamma)(1+\Or(z))\,,\quad -z\notin \rz^+\,,
\ee
and \cite[p.~,139]{Magnus:1966}, $\delta>0$,
\be
\label{largez}
K_0(z)=\sqrt{\frac{\pi}{2z}}\ue^{-z}(1+\Or(z^{-1}))\,, \quad |z|\rightarrow\infty\,,\quad |\arg z|<\frac{3}{2}\pi-\delta\,.
\ee
In addition, we invoke for our analysis the following integral identity, \cite[p.~85]{Magnus:1966}, $x$,~$\xi\in\rz^+$,
\be
\label{integralrep1}
K_{0}(\sqrt{x^2+\xi^2})=\int\limits_{0}^{\infty}\ue^{-x\cosh(t)}\cos(\xi\sinh(t))\ud t\,.
\ee
The following simple lemma will be used.
\begin{lemma}
\label{absfourier}
For $\kappa>0$ and $\eta\in\rz$ we have, $\xi\in\rz$,
\be
\label{Fourierabs}
(F\ue^{-\kappa|\cdot-\eta|})(\xi)=\sqrt{\frac{2}{\pi}}\frac{2\kappa}{\kappa^2+\xi^2}\ue^{-\ui\mu\xi}\,.
\ee
\end{lemma}
\begin{proof}
We first treat the case $\eta=0$:
\be
\label{calc1}
\ba
(F\ue^{-\kappa|\cdot|})(\xi)&=\frac{1}{\sqrt{2\pi}}\int\limits_{\rz}\ue^{\ui x\xi}\ue^{-\kappa |x|}\ud x=\frac{1}{\sqrt{2\pi}}\int\limits_{\rz^+}(\ue^{\ui x\xi}+\ue^{-\ui x\xi})\ue^{-\kappa|x|}\ud x\\
&=\sqrt{\frac{2}{\pi}}\frac{\kappa}{\kappa^2+\xi^2}\,.
\ea
\ee
Now, for $\eta\neq0$ we use the identity $(Ff(\cdot-\mu))(\xi)=(Ff)(\xi)\ue^{-\ui\mu\xi}$ giving \eqref{Fourierabs}.
\end{proof}
\begin{prop}
\label{intmcdonald}
For $y$,~$y'$,~$x_2\in\rz$ we have
\be
\label{intKK}
\ba
&\int\limits_{\rz}K_{0}(k\sqrt{(x_1-y)^2+x_2^2})K_{0}(k\sqrt{(x_1-y')^2+x_2^2})\ud x_1\\
&\vspace{0.25cm}=\frac{2}{\pi k}\int\limits_{\rz}\int\limits_{\rz^+}\int\limits_{\rz^+}\frac{\cosh(t)\cosh(t')\cos(x_2\sinh(t))\cos(x_2\sinh(t'))}{((\cosh(t))^2+\xi^2)((\cosh(t'))^2+\xi^2)}\ue^{-\ui\xi(y-y')} \ud t\ud t'\ud \xi\,.
\ea
\ee
\end{prop}
\begin{proof}
We use Plancherel's theorem and Lemma \ref{absfourier} to calculate
\be
\ba
&\int\limits_{\rz}\ue^{-k|x_1-(y-y')|\cosh(t)}\ue^{-k|x_1|\cosh(t')}\ud x_1\\
&\hspace{0.25cm}=\frac{2}{\pi}\int\limits_{\rz}\frac{k^2\cosh(t)\cosh(t')}{((k\cosh(t))^2+\xi^2)((k\cosh(t'))^2+\xi^2)}\ue^{-\ui\xi(y-y')}\ud \xi\,.
\ea
\ee
Now, we use \eqref{integralrep1} and then Fubini's theorem, \cite[23.7 Corollary]{Bauer:2001}, for an interchange of the $x_1$,~$t$ and $t'$ integration. Moreover, we are going to use the coordinate transformation $\xi\rightarrow k\xi$ giving
\be
\ba
&\int\limits_{\rz}K_{0}(k\sqrt{(x_1-y)^2+x_2^2})K_{0}(k\sqrt{(x_1-y')^2+x_2^2})\ud x_1\\
&\hspace{0.25cm}=\int\limits_{\rz}\int\limits_{\rz^+}\int\limits_{\rz^+}\ue^{-k|x_1-(y-y')|\cosh(t)}\cos(kx_2\sinh(t))\ue^{-k|x_1|\cosh(t')}\cos(kx_2\sinh(t'))\ud t\ud t\ud x_1\\
&\hspace{0.25cm}=\frac{2}{\pi k}\int\limits_{\rz}\int\limits_{\rz^+}\int\limits_{\rz^+}\frac{\cosh(t)\cosh(t')\cos(kx_2\sinh(t))\cos(kx_2\sinh(t'))}{((\cosh(t))^2+\xi^2)((\cosh(t'))^2+\xi^2)}\ue^{-\ui k\xi(y-y')}\ud t\ud t'\ud \xi\,.
\ea
\ee
\end{proof}
We investigate the decay properties of the integral in \eqref{intKK} for large $x_2$. To ease notation we introduce
\be
F(t,t',\xi):=\frac{\cosh(t)\cosh(t')}{((\cosh(t))^2+\xi^2)((\cosh(t'))^2+\xi^2)}\,,
\ee
and we obtain
\begin{lemma}
\label{x2as}
The following estimate 
\be
\ba
&\left|\int\limits_{\rz}\int\limits_{\rz^+}\int\limits_{\rz^+}F(t,t',\xi)\cos(k x_2\sinh(t))\cos(k  x_2\sinh(t'))\ue^{-\ui\xi(y-y')}\ud t\ud t'\ud \xi\right|\\
&\hspace{0.25cm}\leq\frac{1}{(kx_2)^2}\int\limits_{\rz^+}\int\limits_{\rz^+}\frac{\left|(\partial_{t,t'}F(t,t',\xi))\right|}{\cosh(t)\cosh(t')}\ud t\ud t'\ud \xi
\ea
\ee
holds.
\end{lemma}
\begin{proof}
We first observe that $F$, $\partial_{l}F$, $l=t$,~$t'$, and $\partial_{t,t'}F$ are integrable. Then, we treat the $t$ integration with the integration by parts method and arrive at
\be
\ba
&\int\limits_{\rz^+}F(t,t',\xi)\cos(kx_2\sinh(t))\ud t=\int\limits_{\rz^+}\frac{F(t,t',\xi)\cosh(t)}{\cosh(t)}\cos(kx_2\sinh(t))\ud t\\
&\hspace{0.25cm}=\frac{1}{kx_2}\int\limits_{\rz^+}\frac{(\partial_tF(t,t',\xi))}{\cosh(t)}\sin(kx_2\sinh(t))\ud t\,.
\ea
\ee
Doing the same w.r.t. the $t'$ integration and exploiting that $|\sin(\tau)|\leq|\ue^{\ui\tau}|=|\ue^{-\xi(y-y')}|=1$, $\tau\in\rz$, gives
\be
\ba
&|\int\limits_{\rz^+}F(t,t',\xi)\cos(kx_2\sinh(t))\cos(kx_2\sinh(t'))\ue^{-\ui k\xi(y-y')}\ud t\ud t'\ud\xi|\\
&\hspace{0.25cm}=\frac{1}{(kx_2)^2}|\int\limits_{\rz^+}\frac{\partial_{t,t'}F(t,t',\xi)}{\cosh(t)\cosh(t')}\sin(kx_2\sinh(t))\sin(kx_2\sinh(t'))\ue^{-\ui k\xi(y-y')}\ud t\ud t'\ud\xi|\\
&\hspace{0.25cm}\leq\frac{1}{(kx_2)^2}\int\limits_{\rz^+}\frac{\left|\partial_{t,t'}F(t,t',\xi)\right|}{\cosh(t)\cosh(t')}\ud t\ud t'\ud\xi\,.
\ea
\ee
\end{proof}
We are now able to perform an $x_2$-integration in \eqref{intKK}, and we remark that in the following appearing integrals the order of integration is crucial.
\begin{prop}
\label{intmcdonald2}
For $y$,~$y'\in\rz$ we have
\be
\label{intKK1}
\ba
&\int\limits_{\rz}\int\limits_{\rz}K_{0}(k\sqrt{(x_1-y)^2+x_2^2})K_{0}(k\sqrt{(x_1-y')^2+x_2^2})\ud x_1\ud x_2\\
&\vspace{0.25cm}=\int\limits_{\rz}\frac{\pi}{2k^2(1+\xi^2)^{\frac{3}{2}}}\ue^{-\ui k\xi(y-y')}\ud \xi\,.
\ea
\ee
\end{prop}
\begin{proof}
Due to Lemma \ref{x2as} the $x_2$-integral exists. We perform in \eqref{intKK} the substitution $t,t'\rightarrow\asinh(t),\asinh(t')$, use $(\frac{\ud}{\ud t}\sinh(t))=(\cosh(\asinh(t)))^{-1}=(\sqrt{1+t^2})^{-1}$, $\cos(x)\cos(y)=\frac{1}{2}(\cos(x+y)+\cos(x-y))$ and make the substitution $x_2\rightarrow \frac{x_2}{k}$. This gives
\be
\label{intKK2}
\ba
&\int\limits_{\rz}\int\limits_{{\rz^+}^2}\int\limits_{\rz}\frac{\cosh(t)\cosh(t')\cos(kx_2\sinh(t))\cos(kx_2\sinh(t'))}{((\cosh(t))^2+\xi^2)((\cosh(t'))^2+\xi^2)}\ue^{-\ui k\xi(y-y')}\ud \xi \ud t\ud t'\ud x_2\\
&\hspace{0.25cm}=\frac{1}{k}\int\limits_{\rz}\int\limits_{{\rz^+}^2}\int\limits_{\rz}\frac{\cos(x_2(t+t'))+\cos(x_2(t-t'))}{2(t^2+1+\xi^2)(t'^2+1+\xi^2)}\ue^{-\ui k\xi(y-y')}\ud \xi \ud t\ud t'\ud x_2\,.
\ea
\ee
Due to decay property proven in Lemma \ref{x2as} it is not hard to see that we may extend the setting on \cite[p.~36]{MR3328613} to our case. Hence, we may use the $\delta$-identity, \cite[pp.~33,34]{MR3328613},
\be
\label{deltaid}
\frac{1}{2\pi}\int\limits_{\rz}\cos(\xi(x-x_0))\ud \xi=\delta(x-x_0)\,. 
\ee
We arrive at
\be
\ba
&\int\limits_{\rz}K_{0}(k\sqrt{(x_1-y)^2+x_2^2})K_{0}(k\sqrt{(x_1-y')^2+x_2^2})\ud x_1\ud x_2\\
&\hspace{0.25cm}=\int\limits_{\rz^+}\int\limits_{\rz}\frac{2}{k^2(t^2+1+\xi^2)^2}\ue^{-\ui k\xi(y-y')}\ud \xi \ud t\\
&\hspace{0.25cm}=\int\limits_{\rz}\frac{\pi}{2k^2(1+\xi^2)^{\frac{3}{2}}}\ue^{-\ui k\xi(y-y')}\ud \xi\,,
\ea
\ee 
where in the last line we used \cite[3.241~4.]{MR2360010}.
\end{proof}
\subsection{A converse Watson lemma}
To connect the large-$\lambda$ (or large-$k$) asymptotics of the resolvent with the small-$t$ asymptotics of the heat kernel the following lemma will be used, which may be obtained by replacing $a_n$ by $\frac{a_n}{\Gamma(\lambda_n)}$ on \cite[p.~31]{MR1851050}.
\begin{lemma}[Converse Watson Lemma]
\label{lemmaw}
Let $f(\cdot)$ be a continuous function in $(0,\infty)$, $f(t)=0$ for $t<0$, and $\ue^{-c\cdot}f(\cdot)\in L^1(0,\infty)$. Let $F$ be the Laplace transform of $f$, i.e.,
\be
F(z):=\int\limits_{0}^{\infty}f(t)\ue^{-zt}\ud t\,.
\ee
If $F$ possesses the uniform asymptotic expansion
\be
F(z)\sim\sum\limits_{n=0}^{\infty}a_nz^{-\lambda_n}\,,\quad |z|\rightarrow\infty, \ \arg(z-c)\leq\frac{\pi}{2}\,,
\ee
and $\lambda_n\rightarrow\infty$ monotonously as $n\rightarrow\infty$, then
\be
\label{iwl}
f(t)\sim \sum\limits_{n=0}^{\infty}\frac{a_n}{\Gamma(\lambda_n)}t^{\lambda_n-1}\,, \quad t\rightarrow0^+\,.
\ee
\end{lemma} 
\end{document}